\documentclass[runningheads]{llncs}
\usepackage{amssymb}
\usepackage{amsmath}
\usepackage{comment}
\usepackage[disable]{todonotes}
\usepackage[utf8]{inputenc}
\usepackage{shuffle}
\usepackage{multirow}
\usepackage{listings}
\usepackage{mathabx}
\usepackage{ mathrsfs }

\usepackage{hyperref}

\usepackage{diagbox}
\usepackage{slashbox}
\usepackage{tikz}
\usetikzlibrary{arrows,shapes,automata,positioning}
\usetikzlibrary{decorations.pathreplacing,calc}

\newcommand{\stc}{\operatorname{sc}}

\newcommand{\rank}{\operatorname{rk}}
\newcommand{\Syn}{\operatorname{Syn}}
\newcommand{\Cerny}{\v{C}ern\'y}

\usepackage{apxproof}
\theoremstyle{plain}
\newtheoremrep{theorem}{Theorem}
\newtheoremrep{proposition}[theorem]{Proposition}
\newtheoremrep{lemma}[theorem]{Lemma}
\newtheoremrep{claim}[theorem]{Claim}
\newtheoremrep{conjecture}[theorem]{Conjecture}
\newtheoremrep{corollary}[theorem]{Corollary}

\def\dotminus{\mathbin{\ooalign{\hss\raise1ex\hbox{.}\hss\cr
  \mathsurround=0pt$-$}}} 
 
\begin{document}
%

\title{State Complexity of the Set of Synchronizing Words for Circular Automata and Automata over Binary Alphabets}
\titlerunning{The Set of Synchronizing Words for Circular and Binary Automata}

%
%
\author{Stefan Hoffmann\orcidID{0000-0002-7866-075X}}
\authorrunning{S. Hoffmann}
%
\institute{Informatikwissenschaften, FB IV, 
  Universit\"at Trier,  Universitätsring 15, 54296~Trier, Germany, 
  \email{hoffmanns@informatik.uni-trier.de}}
\maketitle              
\begin{abstract}

  Most slowly synchronizing automata over binary alphabets
  are circular, i.e.,
  containing a letter permuting the states in a single cycle, and their set of synchronizing words has maximal state complexity,
  which also implies complete reachability.
  Here, we take a closer look at generalized circular and completely reachable automata.
  We derive that over a binary alphabet every completely
  reachable automaton must be circular, a consequence
  of a structural result stating  that completely
  reachable automata over 
  strictly less letters than states always contain permutational letters.
  We state sufficient conditions for the state complexity
  of the set of synchronizing words of a generalized circular automaton
  to be maximal.
  We apply our main criteria to the family $\mathscr K_n$ 
  of automata that was previously
  only conjectured to have this property. 
  
  %
\end{abstract}

\section{Introduction}

A deterministic semi-automaton is synchronizing if it admits a reset word, i.e., a word which leads to some definite
state, regardless of the starting state. This notion has a wide range of applications, from software testing, circuit synthesis, communication engineering and the like, see~\cite{San2005,Vol2008}.  The famous \v{C}ern\'y conjecture~\cite{Cer64}
states that a minimal synchronizing word has length at most $(n-1)^2$
for an $n$ state automaton.
We refer to the mentioned survey articles~\cite{San2005,Vol2008} for details. 
An automaton is completely reachable, if for each subset of states 
we can find a word which maps the whole state set onto this subset.
This is a generalization of synchronizability,
as a synchronizing word maps the whole state set to a singleton set.
The class of completely reachable automata was formally introduced in~\cite{DBLP:conf/dcfs/BondarV16},
but already in~\cite{DBLP:journals/combinatorics/Don16,DBLP:journals/corr/Maslennikova14}
such automata appear in the results.
The time complexity of deciding if a given automaton is completely reachable
is unknown. A sufficient and necessary criterion for complete reachability of a given automaton
in terms of graphs and their connectivity
is known~\cite{DBLP:conf/dlt/BondarV18}, but it is not known if these graphs
could be constructed in polynomial time.
A special case of the general graph construction, which gives a sufficient
criterion for complete reachability~\cite{DBLP:conf/dcfs/BondarV16},
is known to be constructible in polynomial time~\cite{DBLP:journals/jalc/GonzeJ19}.
The size of a minimal automaton accepting a given regular
language is called the state complexity of that language.
The set of synchronizing words of a given automaton is a regular ideal language
whose state complexity is at most exponential in the size of the original 
automaton~\cite{DBLP:journals/corr/Maslennikova14,DBLP:journals/ijfcs/Maslennikova19}. 
The \v{C}ern\'y family of automata~\cite{Cer64}, a family of automata
yielding the lower bound $(n-1)^2$ for the length of shortest synchronizing words,
is completely reachable, but also the set of synchronizing words has maximal
state complexity~\cite{DBLP:journals/corr/Maslennikova14,DBLP:journals/ijfcs/Maslennikova19}. These properties are shared by many families
of automata that are also slowly synchronizing~\cite{DBLP:conf/mfcs/AnanichevGV10,AnaVolGus2013,DBLP:journals/corr/Maslennikova14,DBLP:journals/ijfcs/Maslennikova19}, i.e., those for which a shortest reset word is close to the 
\Cerny~bound.


\subsubsection{Outline and Contribution:} In Section~\ref{sec::preliminaries} we give definitions
and state known results. Then, in Section~\ref{sec:gen_results}, we give a general criterion for completely
reachable automata to deduce that the set of synchronizing words has maximal state complexity.
We also state a structural result by which we can deduce that in completely reachable automata,
where the number of letters is strictly less than the number of states, we must have permutational letters generating
a non-trivial permutation group. In Section~\ref{sec:gen_circ_aut}, we state sufficient conditions
for generalized circular automata to deduce that their set of synchronizing words
has maximal state complexity.
In Section~\ref{sec:binary_case}, we apply the results from Section~\ref{sec:gen_results}
to deduce
that completely reachable automata over binary alphabets must be circular
and must have a letter mapping precisely two states to a single state. 
Also, with the results from Section~\ref{sec:gen_circ_aut}, we show that the family $\mathscr K_n$, $n > 5$ odd, from~\cite{DBLP:journals/ijfcs/Maslennikova19}
gives completely reachable automata such that the set of synchronizing words has maximal
state complexity. This solves an open problem from~\cite{DBLP:journals/ijfcs/Maslennikova19},
where this was only conjectured. 
The \v{C}ern\'y family
of automata~\cite{Cer64,Vol2008}, the first given family yielding the lower bound $(n-1)^2$ for the length of synchronizing words,
is completely reachable and its set of synchronizing words has maximal state 
complexity~\cite{DBLP:journals/corr/Maslennikova14,DBLP:journals/ijfcs/Maslennikova19}.
These properties are also shared by a wealth of different slowly synchronizing
automata~\cite{DBLP:conf/mfcs/AnanichevGV10,AnaVolGus2013,DBLP:journals/corr/Maslennikova14,DBLP:journals/ijfcs/Maslennikova19}.
Our criteria apply to all the automata mentioned in this previous work.
However, we give an example showing
that our stated conditions are only sufficient, but not necessary.

\section{Preliminaries and Definitions}
\label{sec::preliminaries}


\subsubsection{General Notions:} Let $\Sigma = \{a_1,\ldots, a_k\}$ be a finite set of \emph{symbols}
(also called \emph{letters}), 
 called an \emph{alphabet}. The set $\Sigma^{\ast}$ denotes
the set of all finite sequences, i.e., of all words or strings. The finite sequence of length zero,
or the \emph{empty word}, is denoted by $\varepsilon$. We set $\Sigma^+ = \Sigma^* \setminus \{\varepsilon\}$.
For a given word $w \in \Sigma^*$, we denote by $|w|$
its \emph{length}.
The subsets of $\Sigma^{\ast}$
are called \emph{languages}. 
For $n > 0$ we set $[n] = \{0,\ldots, n-1\}$ and $[0] = \emptyset$.
If $a,b \in \mathbb Z$ and $b \ne 0$, by $a\bmod b$ we denote the unique number $0 \le r < |b|$
with $a = qb + r$ for some $q \in \mathbb Z$.
For some set $X$ by $\mathcal P(X)$ we denote the \emph{power set} of $X$, i.e,
the set of all subsets of $X$. Every function $f : X \to Y$
induces a function $\hat f : \mathcal P(X) \to \mathcal P(Y)$
by setting $\hat f(Z) := \{ f(z) \mid z \in Z \}$.
Here, we will denote this extension also by $f$.
Let $k \ge 1$. A \emph{$k$-subset} $Y \subseteq X$ is a finite set of cardinality~$k$.

\subsubsection{Automata-Theoretic Notions:} A \emph{finite, deterministic} and \emph{complete} automaton will
be denoted by $\mathscr A = (\Sigma, Q, \delta, s_0, F)$
with $\delta : Q \times \Sigma \to Q$ the state transition function, $Q$ a finite set of states, $s_0 \in Q$
the start state and $F \subseteq Q$ the set of final states. 
The properties of being deterministic and complete are implied by the definition of $\delta$
as a total function.
The transition function $\delta : Q \times \Sigma \to Q$
could be extended to a transition function on words $\delta^{\ast} : Q \times \Sigma^{\ast} \to Q$
by setting $\delta^{\ast}(s, \varepsilon) := s$ and $\delta^{\ast}(s, wa) := \delta(\delta^{\ast}(s, w), a)$
for $s \in Q$, $a \in \Sigma$ and $w \in \Sigma^{\ast}$. In the remainder we drop
the distinction between both functions and will also denote this extension by $\delta$.
For $S \subseteq Q$ and $w \in \Sigma^*$, we write $\delta(S, w) = \{ \delta(s, w) \mid s \in S \}$
and $\delta^{-1}(S, w) = \{ q \in Q \mid \delta(q, w) \in S \}$.
The \emph{language accepted} by $\mathscr A = (\Sigma, S, \delta, s_0, F)$ is
$
 L(\mathscr A) = \{ w \in \Sigma^{\ast} \mid \delta(s_0, w) \in F \}.
$
A language $L \subseteq \Sigma^{\ast}$ is called \emph{regular} if $L = L(\mathscr A)$
for some finite automaton $\mathscr A$. 
For a language $L \subseteq \Sigma^{\ast}$ and $u, v \in \Sigma^*$
we define the \emph{Nerode right-congruence} 
with respect to $L$ by $u \equiv_L v$ if and only if
$
 \forall x \in \Sigma : ux \in L \leftrightarrow vx \in L.
$
The equivalence class for some $w \in \Sigma^{\ast}$
is denoted by $[w]_{\equiv L} := \{ x \in \Sigma^{\ast} \mid x \equiv_L w \}$.
A language is regular if and only if the above right-congruence has finite index, and it could
be used to define the minimal deterministic automaton
$
 \mathscr A_L = (\Sigma, Q, \delta, [\varepsilon]_{\equiv_L}, F)
$
with $Q := \{ [w]_{\equiv_L} \mid w \in \Sigma^{\ast} \}$, $\delta([w]_{\equiv_L}, a) := [wa]_{\equiv_L}$
for $a \in \Sigma$, $w \in \Sigma^{\ast}$ and $F := \{ [w]_{\equiv_L} \mid w \in L \}$.
It is indeed the smallest automaton accepting $L$ in terms of the number of states, and we 
will refer to this construction as the minimal automaton~\cite{HopUll79} of $L$.
The \emph{state complexity} of a regular language
is defined as the number of Nerode right-congruence classes.
We will denote this number by $\stc(L)$.
Let $\mathscr A = (\Sigma, Q, \delta, s_0, F)$ be an automaton.
A state $q \in Q$ is \emph{reachable}, if $q = \delta(s_0, u)$
for some $u \in \Sigma^*$. We also say that a state $q$ is reachable from a state $q'$
if $q = \delta(q', u)$ for some $u \in \Sigma^*$.
Two states $q, q'$ are \emph{distinguishable},
if there exists $u \in \Sigma^*$ such that either $\delta(q, u) \in F$ and $\delta(q', u) \notin F$
or $\delta(q, u) \notin F$ and $\delta(q', u) \in F$.
An automaton for a regular language is isomorphic to the minimal automaton
if and only if all states are reachable and distinguishable~\cite{HopUll79}.
A \emph{semi-automaton} $\mathscr A = (\Sigma, Q, \delta)$
is like an ordinary automaton, but without a distinguished start state and without a set of final 
states. Sometimes we will also call a semi-automaton simply an automaton if the context
makes it clear what is meant. Also, definitions without explicit reference
to a start state and a set of final states are also valid for semi-automata.
Let $\mathscr A = (\Sigma, Q, \delta)$ be a finite semi-automaton.
A word $w \in \Sigma^*$ is called \emph{synchronizing} if $\delta(q, w) = \delta(q', w)$
for all $q, q' \in Q$, or equivalently $|\delta(Q, w)| = 1$.
Set $\Syn(\mathscr A) = \{ w \in \Sigma^* \mid |\delta(Q, w)| = 1 \}$.
The \emph{power automaton (for synchronizing words)}
associated to $\mathscr A$ is $\mathcal P_{\mathscr A} = (\Sigma, \mathcal P(Q), \delta, Q, F)$
with start state $Q$, final states $F = \{ \{q \} \mid q\in Q \}$ and the transition function of $\mathcal P_{\mathscr A}$
is the transition function of $\mathscr A$, but applied to subsets of states.
Then, as observed in~\cite{Starke66a},
the automaton $\mathcal P_{\mathscr A}$ accepts the set of synchronizing words,
i.e., $L(\mathcal P_{\mathscr A}) = \Syn(\mathscr A)$.
As for $\{q\} \in F$, we also have $\delta(\{q\}, x) \in F$ for each $x \in \Sigma^*$,
the states in $F$ could all be merged to a single state to get an accepting automaton
for $\Syn(\mathscr A)$.
Also, the empty set is not reachable from $Q$. Hence $\stc(\Syn(\mathscr A)) \le 2^{|Q|} - |Q|$ 
and this bound is sharp~\cite{DBLP:journals/corr/Maslennikova14,DBLP:journals/ijfcs/Maslennikova19}.
We call $\mathscr A$ \emph{completely reachable} if for any non-empty
$S \subseteq Q$ there exists a word $w \in \Sigma^*$ with $\delta(Q, w) = S$, i.e.,
in the power automaton, every state is reachable from the start state. When
we say a \emph{subset of states} in $\mathscr A$ is \emph{reachable}, we mean reachability in $\mathcal P_{\mathscr A}$.
The state complexity of $\Syn(\mathscr A)$ is maximal, i.e., $\stc(\Syn(\mathscr A)) = 2^{|Q|} - |Q|$,
if and only if at least one singleton subset of $Q$
and all subsets $S \subseteq Q$ with $|S| \ge 2$ are reachable, and all non-singleton subsets are distinguishable in $\mathcal P_{\mathscr A}$.
For \emph{strongly connected automata}, i.e., those 
 for which every state is reachable from every other state,
is maximal iff $\mathscr A$ is completely reachable and all $S \subseteq Q$ with $|S| \ge 2$
are distinguishable 
in $\mathcal P_{\mathscr A}$. 
A \emph{permutation} on a finite set $Q$ (which here will always be the set of states of some automaton)
is a bijective function,
a subset of permutations closed under concatenation (and function inversion, but this is implied in the finite case)
is called a \emph{permutation group}.\todo{begriff degree nutze ich nicht.}
The \emph{orbit} of an element from $Q$ under a given permutation group on $Q$
is the sets of all elements to which this element could be mapped by elements from the permutation group.\todo{formaler?}
A permutation group with a single orbit, i.e., every element could be mapped to any other,
is called \emph{transitive}.
A semi-automaton $\mathscr A = (\Sigma, Q, \delta)$ is called \emph{circular},
if some letter acts as a cyclic permutation on all states.
This family\todo{family oder class?} of automata was one of the first inspected with respect to the \Cerny-conjecture~\cite{DBLP:conf/icalp/Pin78}, and the conjecture was finally confirmed for this family~\cite{DBLP:journals/ita/Dubuc96,DBLP:journals/ita/Dubuc98}.
A semi-automaton $\mathscr A = (\Sigma, Q, \delta)$ is called \emph{generalized circular},
if some word acts as a cyclic permutation on all states\footnote{The circular automata are a proper subfamily of the 
generalized circular automata,
as shown by $\mathscr A = (\{a,b\}, [3], \delta)$
with $\delta(0, a) = 1, \delta(1, a) = 0, \delta(2,a) = 2$
and $\delta(0, b) = 0, \delta(1, b) = 2, \delta(2, b) = 1$. The word $ba$
cyclically permutes the states.}.
Let $\mathscr A = (\Sigma, Q, \delta)$ be an automaton and
for $w \in \Sigma^*$ define $\delta_w : Q \to Q$
by $\delta_w(q) = \delta(q,w)$ for all $q \in Q$.
Then, we can associate with $\mathscr A$ the \emph{transformation monoid}
of the automaton $\mathcal T_{\mathscr A} = \{ \delta_w \mid w \in \Sigma^* \}$.
The \emph{rank} of a map $f : Q \to Q$ on a finite set $Q$ is the cardinality of its image.
For a given automaton, seeing a word as a transformation of its state set, the rank of the word
is the rank of this transformation.  
A \emph{permutational letter} is a letter of full rank, i.e., a letter inducing a permutation on the states.

\subsubsection{Known Results:} We will need the following result from~\cite{DBLP:journals/combinatorics/Don16}
to deduce complete reachability of some automata families we consider.

\begin{proposition}[Don~\cite{DBLP:journals/combinatorics/Don16}]
\label{prop:don_compl_reach}
 Let $\mathscr A = (\Sigma, Q, \delta)$
 be a finite circular automaton with $n$ states,
 where $b$ induces a cyclic permutation of the states.
 Suppose we have another letter $a \in \Sigma$
 of rank $n-1$ and
 choose $s,t \in Q$ and $0 < d < |Q|$ such that $\delta(Q, a) = Q \setminus \{s\}$, 
 $|\delta^{-1}(t, a)| = 2$ and $\delta(s, b^d) = t$.
 If $d$ and $n$ are coprime, then
%
%
 for every non-empty set $S \subseteq Q$ of size $k$, there exists
 a word $w_S$ of length at most $n(n-k)$ such that $\delta(Q, w_S) = S$.
\end{proposition}



\section{General Results on the State Complexity of $\Syn(\mathscr A)$}
\label{sec:gen_results}

Our first result states that for completely reachable automata,
to deduce that the set of synchronizing words has maximal state complexity,
we only need to show distinguishability for those subsets of states
with precisely two elements. 

\begin{lemma}[Hoffmann~\cite{DBLP:conf/lata21}]
\label{lem:compl_reach_implies_max_Syn_2sets}
 Let $\mathscr A = (\Sigma, Q, \delta)$ be a completely reachable semi-automaton
 with $n$ states. Then, $\stc(Syn(\mathscr A)) = 2^n - n$
 if and only if all $2$-sets of states are pairwise distinguishable in $\mathcal P_{\mathscr A}$.
\end{lemma} 

With the next result we can deduce information about the structure of completely reachable automata
when the alphabet, or more precisely
only the number of letters of rank $n - 1$, is strictly smaller than the number of states. 
Later, for completely reachable automata over binary alphabet, we can deduce
that they must be circular and have to contain a letter of rank $n - 1$.
Note that we formulate it with a weaker condition than full complete reachability, merely
only with reachability of subsets of size $n - 1$.

\begin{propositionrep}
\label{prop:n-1_reachable_iff_transitive_perm_grp} 
 Let $\mathscr A = (\Sigma, Q, \delta)$ be a semi-automaton 
 with $n$ states, $m$ letters of rank $n - 1$
 and $n > m$.
 Then, the following conditions are equivalent:
 \begin{enumerate}
 \item every subset of size $n - 1$ is reachable,
 \item there exists at least one letter of rank $n - 1$
  and a subset of letters generating a non-trivial permutation group 
  such that every state is in the same orbit as some 
  state not in the image of a rank $n - 1$ letter.
  In particlar, we have at most $m$ orbits.
 \end{enumerate}
\end{propositionrep}
\begin{proof}

   Let us denote the rank of a function $f : [n] \to [n]$ by $\operatorname{rk}(f)$. Note that
 \begin{equation}\label{eqn:rank}
     \rank(fg) \le \min\{ \rank(f), \rank(g) \}
 \end{equation}
 for functions $f, g : [n] \to [n]$.
 Set $n = |Q|$. First, assume that every subset of size $n-1$ is reachable.
 Then some letter, say $a_1$, must have rank $|Q| - 1$.
 For, if not, then any word $w \in \Sigma^*$ 
 has rank $|Q|$, if it is composed out of permutations, or a rank strictly smaller than $|Q| - 1$ , if some letter
 of rank strictly smaller than $|Q| - 1$ appears in it, by Equation~\eqref{eqn:rank}.
 Let $a_1, \ldots, a_m$ be all those letters of rank $n - 1$
 and let $G$ be the subgroup generated by all permutational letters.
 We will employ a slight abuse of notation in the following, we will write,
 for a state $q \in Q$ and $g \in G$,
 simply $\delta(q, g)$ when we actually mean to choose a word $w$ representing
 $g$ and applying this to $q$, i.e., $\delta(q, w)$.
 This poses no problem and is well-defined, 
 as $G$ is defined as a subset of the transformation monoid, which identifies
 words precisely if they define the same transformation. 
 Choose, for each $i \in \{1,\ldots,m\}$,
 a state $s_i \notin \delta(Q, a_i)$, i.e., $Q= \delta(Q, a_i) \cup \{s_i\}$.
 Note that, for $A \subseteq Q$, if $|\delta(A, a_i)| = n - 1$, $i \in \{1,\ldots,m\}$,
 then $\delta(A, a_i) = \delta(Q, a_i)$.
 This yields that for any reachable subset
 $B \subseteq Q$ of size $n - 1$
 we could write $B = \delta(Q, a_i g)$ with $g \in G$ and some $i \in \{1,\ldots,m\}$.
 Hence, 
 $B = Q \setminus \{ \delta(s_i, g) \}$.
 Let $q \in Q$. Then $Q\setminus \{q\}$
 is reachable by assumption. Hence, as the elements in $G$ are permutations, $q = \delta(s_i, g)$
 for some $g \in G$. 
 This yields
 $$ 
  Q = \delta(s_1, G) \cup \ldots \cup \delta(s_m, G)
 $$
 If $G$ is trivial, then $\delta(s_1, G) = \{s_1\}$.
 So, as $n > m$, $G$ must be non-trivial, i.e.,
 we must have some letter of full rank.

 Now, conversely assume the letters of rank $n - 1$
 are $a_1, \ldots, a_m$ for some $m > 0$
 with $\delta(Q, a_i) \cup \{s_i\} = Q$
 for states $s_i \in Q$, $i \in \{1,\ldots,m\}$
 such that
 $$
  Q = \delta(s_1, G) \cup \ldots \cup \delta(s_m, G).
 $$
 Choose some $q \in Q$ and
 suppose $q = \delta(s_i, g)$ for $i \in \{1,\ldots,m\}$
 and $g \in G$.
 Then, as $g$ is a permutation, $Q \setminus \{ q \} = \delta(Q \setminus\{s_i\}, g)$
 and as $Q \setminus \{s_i\} = \delta(Q, a_i)$
 we find $Q \setminus \{q\} = \delta(Q, a_ig)$. ~\qed
\end{proof}

\begin{remark}
 The condition $m < n$ cannot be omitted in Proposition~\ref{prop:n-1_reachable_iff_transitive_perm_grp}.
 For example, let $\mathscr A = (\Sigma, Q, \delta)$
 with $Q = [n]$ and $\Sigma = \{a_1, \ldots, a_n, b_1, \ldots, b_n\}$ be
 such that for $i \in \{1,\ldots,n\}$ we have
 $Q \setminus \{i\} = \delta(Q, a_i)$
 and $a_i$ cyclically permutes $Q \setminus \{i\}$.
 Furthermore, let $b_i$ map some fixed state $q_i \in Q \setminus \{i\}$
 to $\delta(q_i, a_i)$ and act as the identity
 transformation on the rest.
 Then, in $\mathscr A$, even when only using the alphabet $a_1, \ldots, a_n$
 we reach every subset of size $n - 1$. But with the additional
 letters, $\mathscr A$ is also completely reachable, as the subautomaton
 given by $Q \setminus \{i\}$ and only the letters $a_i$ and $b_i$
 equals the \Cerny-automaton, which is completely reachable~\cite{DBLP:journals/ijfcs/Maslennikova19}.
 Hence, combining these facts gives complete reachability of $\mathscr A$, but
 we have no permutational letters at all.
\end{remark}

With this result, we can derive that a completely reachable
automaton whose alphabet is small enough has to contain a non-trivial permutation group as part
of its transformation monoid. Or more specifically, if we only have a single 
letter of rank $n - 1$, this permutation group must be transitive. 

\begin{corollary}
\label{cor:n-1_reachable_iff_transitive_perm_grp} 
 If $\mathscr A = (\Sigma, Q, \delta)$ is completely reachable
 with only a single non-permutational letter 
 and $|Q| > 2$, then $\mathcal T_{\mathscr A}$
 contains a transitive permutation group as a submonoid.
\end{corollary}

If we find a transitive permutation group in the 
transformation monoid of some given automaton, then this automaton
is strongly connected. Hence, let us state the following observation
concerning strongly connected automata.

\begin{lemma}[Hoffmann~\cite{DBLP:conf/lata21}] 
\label{lem:SCC_Syn_max_implies_compl_reach}
 Let $\mathscr A = (\Sigma, Q, \delta)$
 be strongly connected. 
 If $\Syn(\mathscr A)$ has maximal state complexity, then
 $\mathscr A$ is completely reachable.
\end{lemma}

Combining Corollary~\ref{cor:n-1_reachable_iff_transitive_perm_grp}
and Lemma~\ref{lem:SCC_Syn_max_implies_compl_reach}
gives the next lemma.

\begin{lemma}
 Let $\mathscr A = (\Sigma, Q, \delta)$
 be completely reachable with only a single non-permutational letter
 and $|Q| > 2$.
 Then, if $\Syn(\mathscr A)$ has maximal state complexity,
 the semi-automaton is completely reachable.
\end{lemma}

As circular automata are strongly connected, the next follows by
Lemma~\ref{lem:SCC_Syn_max_implies_compl_reach}.

\begin{corollary}
\label{cor:max_sync_implies_compl_reach_binary}
 Let $\mathscr A = (\Sigma, Q, \delta)$ be a circular semi-automaton. If $\stc(\Syn(\mathscr A)) = 2^n - n$,
 then $\mathscr A$ is completely reachable.
\end{corollary}

\section{Generalized Circular Automata}
\label{sec:gen_circ_aut}

Here, Theorem~\ref{prop:bin_max_sc_1} and Proposition~\ref{prop:bin_max_sc}
give sufficient conditions to deduce, for completely reachable circular automata,
that the set of synchronizing words has maximal state complexity.
Both conditions entail all known cases of automata over a binary alphabets for
which the set of synchronizing words has maximal state 
complexity~\cite{DBLP:journals/corr/Maslennikova14,DBLP:journals/ijfcs/Maslennikova19}.
However, at the end of this section, we will show that the stated conditions are not necessary.
In Theorem~\ref{prop:bin_max_sc_1}, we do not assume the automaton
to be completely reachable, but only to be circular and to have a letter of
rank $n - 1$ fulfilling a certain condition.
If we also suppose complete reachability, then the theorem
gives that the set of synchronizing words has maximal state complexity.
Also, note that instead of a letter, any word fulfilling
the mentioned condition in Theorem~\ref{prop:bin_max_sc_1}
will work to give the conclusion. 
Most of the time, we formulate our results for letters, but in all statements
 the assumptions could be formulated with words instead, as the notions of distinguishablity
 do not depend on the length, but only on the existence of certain words\footnote{For, if we choose a finite number of
 words and build the automaton by identifiying these words with new letters, 
 distinguishability or reachability of states (or subsets of states) of this new automaton is inherited to the original automaton. Hence, all results are also valid when stated with words instead of letters, but otherwise
 the same conditions.}.
 However, we have a slight focus on automata over binary alphabets later on, and the results
 of Section~\ref{sec:binary_case} will show that completely reachable automata over
 binary alphabets with at least three states are always circular
 and
 every word that cyclically permutes the states is a power of the cyclic permutation.
 So, we formulate our result with letters instead of words for simplicity.
Intuitively, in Theorem~\ref{prop:bin_max_sc_1}, Equation~\eqref{eqn:first_case_apply_a} says
that we can apply the letter $a$ to reduce the distance modulo $n$ on the cycle
given by $b$, or, by Equation~\eqref{eqn:second_case_apply_a},
that we can map to a state having a specific distance, from which
we can then reduce~it. Please see Figure~\ref{fig:theorem}
for a graphical depiction.

\begin{figure}[htb]
\begin{minipage}{0.8\textwidth}
 \scalebox{.75}{\begin{tikzpicture}
  \tikzstyle{vertex}=[circle,draw,minimum size=25pt,inner sep=1pt]

  \foreach \name/\angle/\text in {S11/0/11, S10/8/10, S5/22/5, S4/30/4, S3/38/3, S2/46/2, S1/54/1, S0/62/0}
    \node[vertex] (\name) at (\angle:10cm) {\small $\text$};

  \path[->] (S0) edge [bend right=45,left] node {$a$} (S3)
            (S0) edge [above,pos=0.5] node {$b$}   (S1)
            (S1) edge [right,pos=0.2] node {$a,b$} (S2)
            (S2) edge [right,pos=0.2] node {$a,b$} (S3)
            (S3) edge [right,pos=0.2] node {$a,b$} (S4)
            (S4) edge [right,pos=0.4] node {$a,b$} (S5)
            (S10) edge [right,pos=0.4] node {$a,b$} (S11);
            
  \node at (9.65, 2.72) {$\vdots$}; 
            
  \node (T0) at (9.6,2.8) {};
  \node (T1) at (9.7,2.3)   {};
  
  \path[->] (S5) edge [right,pos=0.2] node {$a,b$} (T0)
            (T1) edge [right,pos=0.2] node {$a,b$} (S10);
            
 
  \node[align=left] at (4,4) {$q = 0, d = 2$ \\ $\delta(q,b^2a) = \delta(q,a)$ \\ $\delta(q,ba) = \delta(q,ab^{n - 1})$};
  
  \draw [decorate,decoration={brace,amplitude=10pt,mirror,raise=4pt},yshift=0pt]
(10.7,1) -- (10.7,8) node[align=left,midway,xshift=1.9cm]  {Reduce $b$-distance \\ by $d$ according \\ to Equation~\eqref{eqn:first_case_apply_a}.};

   \path[->] (S11) edge [bend left,left] node {$a$} (S5);

   \draw [decorate,decoration={brace,amplitude=10pt,raise=4pt},yshift=0pt]
  (8.5,0) -- (8.5,3.5) node[align=left,midway,xshift=-1.6cm]  {Equation~\eqref{eqn:second_case_apply_a}, \\ map the state \\ to a state of \\ $b$-distance $r=2$ \\ from $\delta(q,a)=3$.};

  
  \node[vertex] (n-1) at (3.2,9) {\small $n-1$}; 
  \node (arc-to-n-1) at (2,9) {};
  \node at (1.8,9) {$\cdots$}; 
  
  \path[->] (arc-to-n-1) edge [above] node {$a,b$} (n-1)
            (n-1)        edge [above] node {$a,b$} (S0);
  
  \node(arc-from-11) at (10,-1.2) {$\vdots$};
  \path[->] (S11) edge [right] node {$b$} (arc-from-11);

\end{tikzpicture}}%
\end{minipage}%
 \caption{\footnotesize Illustration of the conditions stated in Theorem~\ref{prop:bin_max_sc_1}
  for an instance with $d = 2$. Shown are the first twelve states and the state $n-1$
  for a circular automaton with $n$ states. Note that in Theorem~\ref{prop:bin_max_sc_1},
  we suppose $0 < m < n$, and indeed, for $m \in \{0,n\}$ 
  Equation~\eqref{eqn:first_case_apply_a} does not apply in general.}
 \label{fig:theorem}
\end{figure}
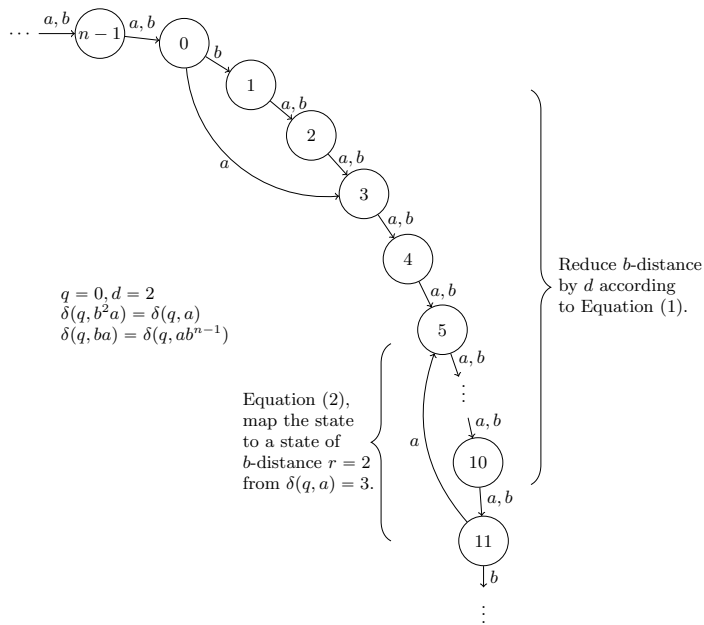

\begin{theoremrep}\todo{Intuition oder Proof Sketch}
\label{prop:bin_max_sc_1}
 Suppose $\mathscr A = (\Sigma, Q, \delta)$
 has $n$ states.
 Let $\{a,b\}\subseteq \Sigma$ (or any two words in $\Sigma^*$).
 Assume the letter $b$ cyclically permutes the states 
 and the letter $a$ has rank $n-1$.
 Then all $2$-sets are distinguishable in $\mathcal P_{\mathscr A}$, if we can find a state $q \in Q$
 and a number $d > 0$ coprime to $n$ such that
 for each $0 < m < n$ we\footnote{Note that $0 < m < n$ implies $\delta(q, b^m) \ne q$. Also
 note that we added $n$ on the right hand side to account for values $d > 1$. In Proposition~\ref{prop:bin_max_sc}
 we only subtract one from the exponent of $b$, which is always non-zero and strictly smaller than $n$, and so we do not needed
 this ``correction for the $b$-cycle'' in case of resulting negative exponents.}
 either have 
 \begin{equation}\label{eqn:first_case_apply_a}
  \delta(q, b^{m}a) = \delta(q, ab^{n + m - d})
 \end{equation} 
 or, but only in case $m$ is not divisible by $d$,
 \begin{equation}\label{eqn:second_case_apply_a}
  \delta(q, ab^r) = \delta(q, b^ma) \mbox{ or } 
  \delta(q, b^mab^r) = \delta(q, a)
 \end{equation}
 for some number $0 \le r < n$ divisible by $d$.
\end{theoremrep}
\begin{proof} 
 Choose the notation for $q \in Q$
 and $r, d > 0$ as in the statement. Suppose $n > 1$.
 Note that $\delta(q, b^da) = \delta(q,a)$, where $\delta(q, b^d) \ne q$, as $d$
 is coprime to $n$. Hence, as $a$ has rank $n-1$, it is injective on $Q \setminus \{q\}$.
 Choose $0 \le s < n$ such that $\delta(q, a) = \delta(q, b^s)$. Then note that,
 as $\delta(q, b^ma) = \delta(q,b^{s+n+m-d})$
 or $\delta(q, b^ma) = \delta(q, b^{s+r})$
 or $\delta(q, b^ma) = \delta(q, b^{s+n-r})$, the set of these equations together with the state $\delta(q, a)$
 completely determine the action of the letter $a$.
 
 Let $\{ p_0, q_0 \}, \{ s_0, t_0 \} \subseteq Q$ be two distinct $2$-sets.
 We want to show that we can distinguish them in $\mathcal P_{\mathscr A}$. 
 First, let $u \in b^*$ be such that $\delta(q_0, u) = q$.
 By applying $u$ to both subsets of $Q$,
 we can assume that $q_0 = q$ with respect to distinguishability. 
 Write $p_0 = \delta(q_0, b^{m_0})$ for some $0 < m_0 < n$.
 Define, for $i \ge 0$, the sequences $\{p_i\}, \{q_i\}, \{s_i\}$ and $\{t_i\}$ by
 $$
  x_{i+1} = \delta(x_i, ab^{n-s}), \quad x \in \{p,q,s,t\}
 $$
 where $p_0, q_o, s_0, t_0$ are given as above.
 Also set $0 \le m_i < n$ with $\delta(q_i, b^{m_i}) = p_i$.
 Then, we have the following.
 
 \begin{quote}
 \noindent\underline{Claim 1: } For all $i \ge 0$ we have $q_i = q$.
 
 {\emph{Proof of Claim 1.}}
   It is $q_0 = q$ by definition. Then, inductively assuming $q_i = q$,
   as $\delta(q, a) = \delta(q, b^s)$,
   we have $q_{i+1} = \delta(q_i, ab^{n-s}) = \delta(q, ab^{n-s}) = \delta(q, b^n) = q$. $\qed$
 \end{quote}
 
 \begin{quote}
 \noindent\underline{Claim 2: } Let $i \ge 0$ and
 suppose $|\{ p_i, q_i \}| = |\{ s_i, t_i\}| = 2$.
 If $p_{i+1} = q_{i+1}$, then $s_i \ne t_i$;
 and if $s_{i+1} = t_{i+1}$, then $p_{i+1}\ne q_{i+1}$.
 
  {\emph{Proof of Claim 2.}}
  As $b$ is a permutation of $Q$ and $a$ has rank $n-1$,
  the transformation given by the word $ab^{n-s}$
  has rank $n-1$. So, as precisely one pair is collapsed,
  we have that if $\{ p_i, q_i \} \ne \{ s_i, t_i \}$
  and both are $2$-subsets, then $\{ p_{i+1}, q_{i+1} \} \ne \{ s_{i+1}, t_{i+1} \}$,
  as only $|\{ p_{i+1}, q_{i+1}, s_{i+1}, t_{i+1} \}| \in \{2,3,4\}$ is possible
  because at most one $2$-set could be collapsed.
  Note that, as by assumption $ \{ p_0, q_0 \} \ne \{ s_0, t_0 \}$,
  this gives inductively $\{ p_i, q_i \} \ne \{ s_i, t_i \}$
  for all $i$ with $|\{ p_i, q_i \}| = |\{ s_i, t_i\}| = 2$.
  As $|\{ p_i, q_i \}| = |\{ s_i, t_i\}| = 2$ is assumed in the claim,
  we have $\{ p_i, q_i \} \ne \{ s_i, t_i \}$
  and as at most one pair could be collapsed the claim follows. 
  $\qed$
 \end{quote}
 
  \begin{quote}
 \noindent\underline{Claim 3: } There exists some $i > 0$
  such that $p_i = q_i = q$.
 
  {\emph{Proof of Claim 3.}} First, suppose we have 
  some $i \ge 0$ and $|\{ p_i, q_i \}| = 2$ 
  such that, after reading the letter $a$, Equation~\eqref{eqn:second_case_apply_a}
  applies, i.e., we have
  $$
   \delta(q_i, ab^r) = \delta(q_i, b^{m_i}a) \mbox{ or }
   \delta(q_i, b^{m_i}ab^r) = \delta(q_i, a).
  $$
  Write $r = m d$ for some $m > 0$.
  \begin{enumerate}
  \item[(i)] Suppose $\delta(q_i, ab^r) = \delta(q_i, b^{m_i}a) = \delta(p_i, a)$.
    
    Then $p_{i+1} = \delta(p_i, ab^{n-s}) = \delta(q_i, ab^{n-s+r}) = \delta(q_{i+1}, b^r) = \delta(q_{i}, b^r)$,
    using $q_i = q$, by Claim 1, and $\delta(q, ab^{n-s}) = q$.
    Then $\delta(p_{i+1}, a) = \delta(q, b^ra)$.
    As $r \equiv 0 \pmod{n}$, Equation~\eqref{eqn:first_case_apply_a}
    must apply, and, as $0 < r < n$, we find
    $$
     \delta(p_{i+1}, a) = \delta(q, ab^{n + r-d}),
    $$
    which equals $\delta(q, ab^{r - d})$ as $r - d \ge 0$.
    So, $p_{i+2} = \delta(p_{i+1}, ab^{n-s}) = \delta(q, ab^{n-s + r - d}) = 
    \delta(\delta(q, ab^{n-s}, b^{r - d})) = \delta(q, b^{r-d})$.
    Continuing inductively for $m-1$ additional steps, we find
    $$
     p_{i+1+m} = \delta(q, b^{r - md}) = q = q_{i+1+m}.
    $$
    
  \item[(ii)] Suppose $\delta(q_i, b^{m_i}ab^r) = \delta(q_i, a)$.
   
   As $p_i = \delta(q_i, b^{m_i})$, we have $\delta(p_i, ab^r) = \delta(q_i, a)$.
   Hence, for $0 \le j < n$ 
   such that $\delta(p_i, ab^{j}) = q$,
   we have $\delta(q_i, ab^j) = \delta(\delta(p_i, ab^r), b^j) = \delta(\delta(p_i, ab^j), b^r) = \delta(q, b^r)$.
   Then we can write $r = md$ and proceed exactly 
   as in Case (i), to find
   $$
    \delta(\{q, \delta(q, b^r)\}, (ab^{n-s})^m) = \{q\}.
   $$
   So that $\delta(\{ \delta(p_i, a), \delta(p_i, ab^r) )\}, b^j(ab^{n-s})^m)) = \{q\}$.
   
  \end{enumerate}
  Otherwise, after reading the letter $a$ in the process of constructing the sequences $p_i$
  and $q_i$, only Equation~\eqref{eqn:first_case_apply_a}
  applies. As $d$ is coprime to $n$, we find $k \ge 0$ such that $m - kd \equiv 0 \pmod{n}$, i.e.,
  $m - kd + ln = 0$ for some $k \ge l$ (as $0 \le m < n$, we have $l \le k$,
  for $l > k$ would imply $ln > kd$, and so we could not have $ln + m = kd$). But note
  that $l \le 0$ is possible.
  By Equation~\eqref{eqn:first_case_apply_a} and using $k \ge l$,
  \begin{align*}
   p_k & = \delta(p_0, (ab^{n-s})^k) \\
       & = \delta(q_0, b^m(ab^{n-s})^k) & [p_0 = \delta(q_0, b^m)]\\
       & = \delta(q, b^m(ab^{n-s})^k)  & \mbox{[$q_0 = q$]} \\
       & =  \delta(q, b^mab^{n-s}(ab^{n-s})^{k-1}) \\ 
       & = \delta(q, ab^{n + m - d}b^{n-s}(ab^{n-s})^{k-1}) & \mbox{[Equation~\eqref{eqn:first_case_apply_a}]} \\
       & = \delta(q, b^{n + m - d}(ab^{n-s})^{k-1}) & \mbox{[$\delta(q, ab^{n-s}) = q$]} \\
       & = \delta(q, ab^{2n + m - 2d}b^{n-s}(ab^{n-s})^{k-2}) \\
       & \quad \quad \vdots \\
       & = \delta(q, b^{(k-1)n + m - (k-1)d}ab^{n-s}) \\
       & = \delta(q, b^{kn + m - kd}) \\ 
       & = \delta(q, b^{kn - ln}) \\
       & = \delta(q, b^{(k-l)n}) = q = q_k.
   \end{align*}
   So, the word $w = (ab^{n-s})^k$ maps both states $\{ p_0, q_0 \}$
   to the same state~$q$.~$\qed$
 \end{quote}

 The Claims (2) and (3) imply that we have some word $w \in (ab^{n-s})^*$
 which collapses precisely one $2$-set. 
 For, by Claim (3), we must have some smallest $i > 0$
 such that $|\{ p_{i-1}, q_{i-1}\}| = |\{s_{i-1}, t_{i-1}\}| = 2$
 and either $p_i = q_i$, or $s_i = t_i$, but not both could
 be equal by Claim (2) and the minimality of $i$.~$\qed$

\end{proof}

In the formulation of Theorem~\ref{prop:bin_max_sc_1}, we have $r > 0$,
as in this case $m \ne 0$. Also, note that $\delta(q, b^mab^r) = \delta(q,a)$
is equivalent with $\delta(q,ab^{n-r}) = \delta(q, b^ma)$,
as for any states $s,t \in Q$ and $0 \le k < n$ we have $\delta(s, b^k) = t$
if any only if $\delta(t, b^{n-k}) = s$, as $\delta(s, b^n) = s$.
The conditions mentioned in Theorem~\ref{prop:bin_max_sc_1}
are the most general ones in this paper, but let us state next, as
a corollary, a more relaxed formulation, stating that we can reduce the distance
on the cycle by one for each application of some word of rank $n - 1$.

\begin{corollary}\todo{automaton semi-automaton nicht unterscheiden, schreiben.}
\label{cor:bin_max_sc_1}
 Let $\mathscr A = (\Sigma, Q, \delta)$
 be a circular automaton with $n$ states
 where the letter $b$ permutes the states with a single orbit.
 Suppose we find a word $w \in \Sigma^*$ 
 and state $q \in Q$
 such that, for\footnote{
 Note that here, even if the bounds for $m$ from Theorem~\ref{prop:bin_max_sc_1}
 do not include this case, $\delta(q, w) = \delta(q, b^nw) = \delta(q, wb^{n-1})$,
 which is equivalent with $\delta(q, w) = \delta(q, b)$.}
 $0 \le m < n$,\todo{für einheitlichkeit m und m - 1 schreiben.}
 \begin{equation}\label{eqn:first_case_apply_w}
  \delta(q, b^{m + 1}w) = \delta(q, wb^{m})
 \end{equation}
 Then all $2$-sets are distinguishable in $\mathcal P_{\mathscr A}$.
 In particular, if $\mathscr A$ is completely reachable,
 then $\stc(\Syn(\mathscr A)) = 2^n - n$.
\end{corollary}
\begin{proof}
 Set $s = \delta(q, b)$ and $t = \delta(q, w)$.
 Then, $\delta(s, w) = t = \delta(q, w)$ and $s \ne q$.
 For $m \in \{1,\ldots,n-1\}$
 we have $\delta(q, b^m w) = \delta(t, b^{m-1})$
 and 
 \[ 
 \{ \delta(q, b), \delta(q, b^2), \ldots, \delta(q, b^{n-1}) \} = Q \setminus \{q\}.
 \]
 So, as $b$ is a permutation, $w$ acts injective on $Q \setminus \{q\}$
 and has rank $n - 1$.
 Now, apply Theorem~\ref{prop:bin_max_sc_1}, 
 interpreting $w$ as the letter $a$ of rank $n - 1$. \qed
\end{proof}

Actually, for the relaxed condition mentioned in Corollary~\ref{cor:bin_max_sc_1}
we can give a small strengthening by only requiring that we can reduce the ``cyclic distance''
for all states which are no more than $\lfloor n/2 \rfloor + 1$ steps, or applications
of $b$, away from some
specific state. 

\begin{propositionrep}
\label{prop:bin_max_sc}
 Let $\Sigma = \{a,b\}$ and suppose $\mathscr A = (\Sigma, Q, \delta)$
 has $n$ states and is completely reachable 
 with the letter $a$ having rank $n-1$ and the letter $b$ permuting the states with a single orbit. 
 Then $\stc(\Syn(\mathscr A)) = 2^n - n$
 if we can find a state $q \in Q$
 such that for all $0 \le m \le \lfloor n/2 \rfloor - 1$ we have 
 \begin{equation}\label{eqn:binary_max_sync_states}
   \delta(q, b^{m+1}a) = \delta(q, ab^m).
 \end{equation}
\end{propositionrep}
\begin{proof} 
 Let $q \in Q$ be the state from the statement.
 Note that, as $\delta(q, ba) = \delta(q, a)$ and $a$ has rank $n - 1$,
 the letter $a$ acts injective on $Q \setminus \{ q \}$. Also,
 on all states $\delta(q, b^k)$ with $0 \le k \le \lfloor n / 2 \rfloor$
 the action of $a$ is determined by the state $\delta(q,a)$.
 We will show that all $2$-sets of states are distinguishable in $\mathcal P_{\mathscr A}$.
 By Lemma~\ref{lem:compl_reach_implies_max_Syn_2sets}, this will give
 our claim.
 Let $\{s,t\}, \{p,r\}\subseteq Q$ be two distinct $2$-sets.
 Choose $m_1, m_2 > 0$ minimal such that
 \begin{equation}\label{eqn:m1}
  \delta(s, b^{m_1}) = t \mbox{ or } \delta(t, b^{m_1}) = s
 \end{equation}
 and
 \begin{equation}\label{eqn:m2}
  \delta(p, b^{m_2}) = r \mbox{ or } \delta(r, b^{m_2}) = p.
 \end{equation}
 As $\delta(t, b^{m_1}) = s$ if and only if $\delta(s, b^{n - m_1}) = t$, and similarly for $\{p,r\}$,
 we have $0 < m_1, m_2 \le \lfloor n / 2 \rfloor$.
 We will do induction on $\min\{m_1, m_2\}$.
 Without loss of generality, assume $m_1 \le m_2$
 and $\delta(s, b^{m_1}) = t$.
 Choose $0 \le k < n$ such that $\delta(s, b^k) = q$.
 Then, by assumption, $\delta(s, b^{k+m_1}a) = \delta(s, b^k a b^{m_1-1})$.
 We distinguish two cases.
 \begin{enumerate}
 \item[(i)] Suppose $s \notin \{p, r\}$. Then $q \notin \delta(\{ p, r \}, b^k)$. As $a$ acts injective on $Q \setminus \{q\}$,
   we have $|\delta(\{p,r\}, b^ka)| = 2$.
   If $m_1 = 1$,then 
   $$
    \delta(s, b^k a) = \delta(s, b^k a b^{m_1 - 1}) = \delta(s, b^{k+m_1}a) = \delta(t, b^k a).
   $$
   Hence $|\delta(\{s,t\}, b^k a)| = 1$ and the sets $\{s,t\}$ and $\{p,r\}$
   are distinguished in $\mathcal P_{\mathscr A}$ by $b^k a$.
   If $m_1 > 1$, then
   $\delta(s, b^kab^{m_1-1}) = \delta(s, b^{k+m_1}a) = \delta(t, b^k a)$.
   As $0 < m_1 - 1 < \lfloor n / 2 \rfloor$, we have $| \{ \delta(s, b^k a), \delta(t, b^k a) \} | = 2$.
   Hence, for the two $2$-sets
   $$
    \{ \delta(s, b^k a), \delta(t, b^k a) \} \mbox{ and } \{\delta(p, b^k a), \delta(r, b^k a) \}
   $$
   the minimal powers of $b$ that map them to each other, i.e., fulfill
   the corresponding Equations~\eqref{eqn:m1} and~\eqref{eqn:m2},
   have strictly smaller exponents than $\min\{m_1, m_2\}$.
   So, we can use our induction hypothesis, implying that some word $u \in \Sigma^*$
   maps one set to a singleton set, but not the other.
   Then $b^k a u$ would distinguish $\{s,t\}$ and $\{p,r\}$.
   
 \item[(ii)] Suppose $s \in \{p, r\}$. Without loss of generality, assume $\delta(p, b^{m_2}) = r$.
  If $s = p$, then $m_1 = m_2$ would imply $r = t$, which is excluded as both $2$-sets are assumed
  to be distinct. Hence, as $m_1 = \min\{m_1, m_2\}$, we have $m_1 < m_2$.
  If $m_1 = 1$, then, as in case (i), we have $|\delta(\{s,t\}, b^k a)| = 1$,
  but 
  $$
   \delta(q, b) \notin \{ \delta(p, b^k), \delta(r, b^k) \}. 
  $$
  As $\delta(q, b) = \delta(r, b^k)$ would imply, as $b$ is a permutation,
  that $\delta(p, b) = r$.
  But, as $p = s$, $\delta(p, b^k) = q$, and $\delta(p, b^{m_2}) = r$ 
  is minimal with $1 < m_2 < \lfloor n/2 \rfloor$.
  Hence $\delta(p, b) = r$ would contradict the minimality of $m_2$.
  So, as $a$ acts injective on $Q \setminus \{ \delta(q, b) \}$, as it only
  collapses $\{q, \delta(q,b)\}$, we have $|\{ \delta(p, b^k), \delta(r, b^k) \}| = 2$.
  So, the word $b^k a$ distinguishes $\{s,t\}$ and $\{p,r\}$.
  Now suppose $m_1 > 1$.
  Then, as in case (i), $\delta(s, b^k a b^{m_1 - 1}) = \delta(t, b^k a)$
  and $| \{ \delta(s, b^k a), \delta(t, b^k a) \} | = 2$.
  Similarly, as $1 < m_2 < n$, we get $| \{ \delta(p, b^k a), \delta(r, b^k a) \} | = 2$.
  But, again, the minimal powers of $b$ that map them to each other, i.e., fulfill
   the corresponding Equations~\eqref{eqn:m1} and~\eqref{eqn:m2},
   have strictly smaller exponents than $\min\{m_1, m_2\}$.
   So, we can use our induction hypothesis, implying that some word $u \in \Sigma^*$
   maps one set to a singleton set, but not the other.
   Then $b^kau$ would distinguish $\{s,t\}$ and~$\{p,r\}$.
   If we have $s = r$, then $\delta(r, b^{n - m_2}) = p$.
   Similarly to the case $s = p$, we have $m_1 < n - m_2$.
   As $m_1 < n - m_2 < n$, we can argue like in the previous case $s = p$
   to find either two new $2$-sets with a stricly smaller induction parameter,
   or one is a singleton, but not the other.
   Note that for $1 < m < n$ we have $|\delta(\{q, \delta(q, b^m)\}, a)| = 2$.
 \end{enumerate} 
 So, by induction, we can distinguish all $2$-sets in $\mathcal P_{\mathscr A}$.~\qed

\end{proof} 

Finally, we show that the mentioned sufficient conditions
are not necessary. In Example~\ref{ex:thm_not_usable}
we will give a circular automaton whose set of synchronizing words has maximal state complexity
but for which this could not be derived with any of the results stated here.

\begin{example} 
\label{ex:thm_not_usable}

 Let $\mathscr A = (\Sigma, [4], \delta)$ with $\Sigma = \{a,b\}$,
 $\delta(i, b) = (i+1) \bmod 4$
 and 
 $ 
  \delta(0, a) = 1, \delta(1, a) = 2, \delta(2, a) = 1, \delta(3, a) = 3.
 $
 Please see Figure~\ref{fig:thm_not_usable} for a graphical depiction of $\mathscr A$
 and $\mathcal P_{\mathscr A}$.
 Then, all words of rank $3$ 
 are listed in Table~\ref{tab:rank_3_words}.
 In each word $w$ of rank $3$ the distance of the two distinct states mapped
 to one state is $2$. So, in Equation~\eqref{eqn:first_case_apply_a},
 for each such word of rank $3$ (in place of $a$),
 we would have $d = 2$. But $2$ is not coprime to $4$, hence Theorem~\ref{prop:bin_max_sc_1}
 does not apply here. However, we have $\stc(\Syn(\mathscr A)) = 2^4 - 4 = 12$.
 We see in Figure~\ref{fig:thm_not_usable}
 that every subset is reachable.
 We also see that $a$ distinguishes $\{0,2\}$ from every other $2$-set
 of states, $ba$ distinguishes $\{1,3\}$ from every other, $baba$ distinguishes $\{2,3\}$
 from $\{0,3\}$, $\{1,2\}$ and $\{0,1\}$ and these latter three $2$-sets
 are easily seen to be distinguishable by words in $b^*aba$.
 So, by Lemma~\ref{lem:compl_reach_implies_max_Syn_2sets},
 all non-empty subsets of states are distinguishable.

 \begin{table}[ht]
  \vspace{-0.7cm}
  \centering
 \[
  \begin{array}{|cl|cl|cl|cl|} 
   \hline
   \mbox{Word}        & \mbox{Mapping} & \mbox{Word} & \mbox{Mapping} & \mbox{Word} & \mbox{Mapping} & \mbox{Word}   & \mbox{Mapping} \\ \hline
   b^2 a^2 b^2 & [0, 1, 0, 3]   & a        & [1,2,1,3]      & ab^3        & [0,1,0,2]  & ab^2ab    & [0,2,0,3]  \\  
   a^2 b^2     & [0, 3, 0, 1]   & b^2a     & [1,3,1,2]      & b^2ab^3     & [0,2,0,1]  & ab^2a^2b  & [0,3,0,2]  \\
   ab^2a^2b^2  & [1, 0, 1, 3]   & a^2      & [2,1,2,3]      & a^2b^3      & [1,0,1,2]  & b^2ab     & [2,0,2,3]  \\
   ab^2ab^2    & [1, 3, 1, 0]   & b^2a^2   & [2,3,2,1]      & b^2a^2b^3   & [1,2,1,0]  & ab        & [2,3,2,0]  \\
   ab^2        & [3, 0, 3, 1]   & ab^2a    & [3,1,3,2]      & ab^2ab^2    & [2,0,2,1]  & b^2a^2b   & [3,0,3,2]  \\ 
   b^2ab^2     & [3, 1, 3, 0]   & ab^2a^2  & [2,2,3,1]      & ab^2a^2b^3  & [2,1,2,0]  & a^2b      & [3,2,3,0]  \\ \hline
   bab^2a^2b^2 & [0,1,3,1]      & ba^2     & [1,2,3,2]      & ba^2b^3     & [0,1,2,1]  & b^3ab     & [0,2,3,2]  \\
   bab^2       & [0,3,1,3]      & bab^2a   & [1,2,3,2]      & bab^2ab^3   & [0,2,1,2]  & b^3a^2b   & [0,3,2,3]  \\
   b^3a^2b^2   & [1,0,3,0]      & ba       & [2,1,3,1]      & bab^3       & [1,0,2,0]  & bab^2ab   & [2,0,3,1]  \\
   b^3ab^2     & [1,3,0,3]      & bab^2a^2 & [2,3,1,3]      & bab^2a^2b^3 & [1,2,0,2]  & ba^2b     & [2,3,0,3]  \\
   ba^2b^2     & [3,0,1,0]      & b^3a     & [3,1,2,1]      & b^3ab^3     & [2,0,1,0]  & bab^2a^2b & [3,0,2,1]  \\
   bab^2ab^2   & [3,1,0,1]      & b^3a^2   & [3,2,1,3]      & b^3a^2b^3   & [2,1,0,2]  & bab       & [3,2,0,2] \\ \hline
  \end{array}
 \]
   \caption{All rank $3$ words for the automaton from Example~\ref{ex:thm_not_usable}. To the right of each word the induced
   transformation on the states is written, where $j \in [4]$ written at position $i \in [4]$
   means the word maps the state $i$ to state $j$. The entries are ordered such that for two words $u,v$
   in the same row we have $\delta(i, u) = \delta(j, u)$ iff $\delta(i, v) = \delta(j, v)$ for $i \in [4]$
   and the images of words in the same column are equal.} 
  \label{tab:rank_3_words}
 \end{table}
 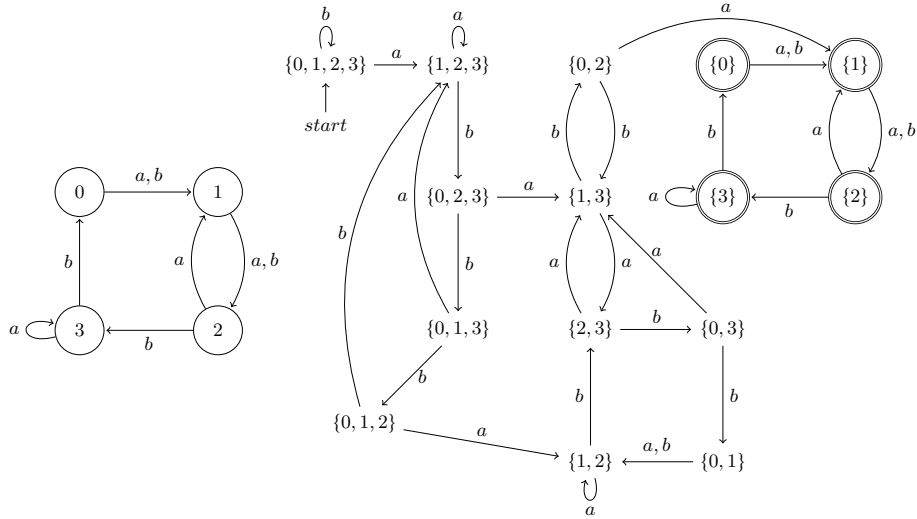
\begin{figure}[htb]
\begin{minipage}{0.3\textwidth}
 \scalebox{.8}{\begin{tikzpicture}[shorten >=1pt,->,node distance=2.3cm]
  
  \node[state] (0) {$0$};
  \node[state, right of=0] (1) {$1$};
  \node[state, below of=1] (2) {$2$};
  \node[state, left of=2]  (3) {$3$};
  
  \path[->] (0) edge [above] node {$a,b$} (1)
            (1) edge [bend left, right] node {$a,b$} (2)
            (2) edge [below] node {$b$} (3)
            (3) edge [left]  node {$b$} (0)
            (2) edge [bend left, left] node {$a$} (1)
            (3) edge [loop left] node {$a$} (3);
\end{tikzpicture}}%
\end{minipage}%
\begin{minipage}{0.7\textwidth}
\scalebox{.8}{\begin{tikzpicture}[shorten >=1pt,->,node distance=2.2cm]
  \tikzstyle{vertex}=[]
  
   \node at (0,-1) (s) {$start$};

  \node[vertex] (0123)               {$\{0,1,2,3\}$};
  \node[vertex, right of=0123] (123) {$\{1,2,3\}$};
  \node[vertex, below of=123]  (023) {$\{0,2,3\}$};
  \node[vertex, below of=023]   (013) {$\{0,1,3\}$};
  \node[vertex, below left of=013]  (012) {$\{0,1,2\}$};
  
  \node[vertex, right of=023] (13) {$\{1,3\}$};
  \node[vertex, above of=13]  (02) {$\{0,2\}$};
  \node[vertex, below of=13]  (23) {$\{2,3\}$};
  \node[vertex, right of=23]  (03) {$\{0,3\}$};
  \node[vertex, below of=03]  (01) {$\{0,1\}$};
  \node[vertex, left of=01]  (12) {$\{1,2\}$};
  
  \node[state, accepting, right of=02] (0) {$\{0\}$};
  \node[state, accepting, right of=0]  (1) {$\{1\}$};
  \node[state, accepting, below of=1]  (2) {$\{2\}$};
  \node[state, accepting, left of=2]   (3) {$\{3\}$};

  \path (123) edge [loop above] node {$a$} (123);
  \path (012) edge [above] node {$a$} (12);
  \path (013) edge [bend left,left] node {$a$} (123);
  
  \path[->] (s) edge (0123);
  
    \path[->] (0) edge [above] node {$a,b$} (1)
            (1) edge [bend left, right] node {$a,b$} (2)
            (2) edge [below] node {$b$} (3)
            (3) edge [left]  node {$b$} (0)
            (2) edge [bend left, left] node {$a$} (1)
            (3) edge [loop left] node {$a$} (3);

  \path[->] (0123) edge [above] node {$a$} (123);
  \path[->] (0123) edge [loop above] node {$b$} (0123);
  \path[->] (123)  edge [right] node {$b$} (023); 
  \path[->] (023)  edge [right] node {$b$} (013);
  \path[->] (013)  edge [right] node {$b$} (012);
  \path[->] (012)  edge [bend left,left] node {$b$} (123);
  
  \path[->] (023)  edge [above] node {$a$} (13);
  \path[->] (13)   edge [bend left,left] node {$b$} (02);
  \path[->] (02)   edge [bend left,right] node {$b$} (13);
  \path[->] (13)   edge [bend left,right] node {$a$} (23);
  \path[->] (23)   edge [bend left,left] node {$a$} (13);
  \path[->] (23)   edge [above] node {$b$} (03);
  \path[->] (03)   edge [right] node {$b$} (01);
  \path[->] (01)   edge [above] node {$a,b$} (12);
  \path[->] (12)   edge [left] node {$b$} (23);
  \path[->] (03)   edge [above] node {$a$} (13);
  \path[->] (12)   edge [loop below] node {$a$} (12);
  
  \path[->] (02) edge [bend left,above] node {$a$} (1);
  
\end{tikzpicture}}%
\end{minipage}%
 \caption{\footnotesize The automaton from Example~\ref{ex:thm_not_usable}
 and its power automaton. An example of an automaton whose
 set of synchronizing words has maximal state complexity
 but for which Theorem~\ref{prop:bin_max_sc_1}
 or Proposition~\ref{prop:bin_max_sc}
 do not apply, not for $a$ and not for any word of rank~$3$. The final states in the
 power automaton are marked with double circles.}
 \label{fig:thm_not_usable}
\end{figure}
 \end{example}

\begin{toappendix}
 Some remarks how all the words of rank $3$
  were generated in Example~\ref{ex:thm_not_usable}, to convince the reader that
 it is indeed correct. I used \textsc{GAP}\footnote{\url{https://www.gap-system.org/}}
 to compute all words of rank $3$ with the following commands.
 
 {\footnotesize
 \begin{verbatim}
gap> LoadPackage( "SgpViz" );;
gap> S:=Semigroup(Transformation([2,3,4,1]),Transformation([2,3,2,4]));;
gap> Print(DotForDrawingDClassOfElement(S,Transformation([2,3,2,4])));
digraph  DClassOfElement {
graph [center=yes,ordering=out];
node [shape=plaintext];
edge [color=cornflowerblue,arrowhead=none];
1 [label=<
<TABLE BORDER="0" CELLBORDER="0" CELLPADDING="0" CELLSPACING="0" PORT="1">
<TR><TD BORDER="0"><TABLE CELLSPACING="0"><TR><TD BGCOLOR="white" BORDER="0">*b^2a^2b^2</TD></TR>
<TR><TD BGCOLOR="white" BORDER="0">a^2b^2</TD></TR>
<TR><TD BGCOLOR="white" BORDER="0">ab^2a^2b^2</TD></TR>
<TR><TD BGCOLOR="white" BORDER="0">ab^2ab^2</TD></TR>
<TR><TD BGCOLOR="white" BORDER="0">ab^2</TD></TR>
<TR><TD BGCOLOR="white" BORDER="0">b^2ab^2</TD></TR>
</TABLE></TD><TD BORDER="0"><TABLE CELLSPACING="0"><TR><TD BGCOLOR="white" BORDER="0">a</TD></TR>
<TR><TD BGCOLOR="white" BORDER="0">b^2a</TD></TR>
<TR><TD BGCOLOR="white" BORDER="0">*a^2</TD></TR>
<TR><TD BGCOLOR="white" BORDER="0">b^2a^2</TD></TR>
<TR><TD BGCOLOR="white" BORDER="0">ab^2a</TD></TR>
<TR><TD BGCOLOR="white" BORDER="0">ab^2a^2</TD></TR>
</TABLE></TD><TD BORDER="0"><TABLE CELLSPACING="0"><TR><TD BGCOLOR="white" BORDER="0">ab^3</TD></TR>
<TR><TD BGCOLOR="white" BORDER="0">b^2ab^3</TD></TR>
<TR><TD BGCOLOR="white" BORDER="0">a^2b^3</TD></TR>
<TR><TD BGCOLOR="white" BORDER="0">b^2a^2b^3</TD></TR>
<TR><TD BGCOLOR="white" BORDER="0">ab^2ab^3</TD></TR>
<TR><TD BGCOLOR="white" BORDER="0">ab^2a^2b^3</TD></TR>
</TABLE></TD><TD BORDER="0"><TABLE CELLSPACING="0"><TR><TD BGCOLOR="white" BORDER="0">ab^2ab</TD></TR>
<TR><TD BGCOLOR="white" BORDER="0">ab^2a^2b</TD></TR>
<TR><TD BGCOLOR="white" BORDER="0">b^2ab</TD></TR>
<TR><TD BGCOLOR="white" BORDER="0">ab</TD></TR>
<TR><TD BGCOLOR="white" BORDER="0">b^2a^2b</TD></TR>
<TR><TD BGCOLOR="white" BORDER="0">a^2b</TD></TR>
</TABLE></TD></TR>
<TR><TD BORDER="0"><TABLE CELLSPACING="0"><TR><TD BGCOLOR="white" BORDER="0">bab^2a^2b^2</TD></TR>
<TR><TD BGCOLOR="white" BORDER="0">bab^2</TD></TR>
<TR><TD BGCOLOR="white" BORDER="0">b^3a^2b^2</TD></TR>
<TR><TD BGCOLOR="white" BORDER="0">b^3ab^2</TD></TR>
<TR><TD BGCOLOR="white" BORDER="0">ba^2b^2</TD></TR>
<TR><TD BGCOLOR="white" BORDER="0">bab^2ab^2</TD></TR>
</TABLE></TD><TD BORDER="0"><TABLE CELLSPACING="0"><TR><TD BGCOLOR="white" BORDER="0">ba^2</TD></TR>
<TR><TD BGCOLOR="white" BORDER="0">bab^2a</TD></TR>
<TR><TD BGCOLOR="white" BORDER="0">ba</TD></TR>
<TR><TD BGCOLOR="white" BORDER="0">bab^2a^2</TD></TR>
<TR><TD BGCOLOR="white" BORDER="0">b^3a</TD></TR>
<TR><TD BGCOLOR="white" BORDER="0">b^3a^2</TD></TR>
</TABLE></TD><TD BORDER="0"><TABLE CELLSPACING="0"><TR><TD BGCOLOR="white" BORDER="0">*ba^2b^3</TD></TR>
<TR><TD BGCOLOR="white" BORDER="0">bab^2ab^3</TD></TR>
<TR><TD BGCOLOR="white" BORDER="0">bab^3</TD></TR>
<TR><TD BGCOLOR="white" BORDER="0">bab^2a^2b^3</TD></TR>
<TR><TD BGCOLOR="white" BORDER="0">b^3ab^3</TD></TR>
<TR><TD BGCOLOR="white" BORDER="0">b^3a^2b^3</TD></TR>
</TABLE></TD><TD BORDER="0"><TABLE CELLSPACING="0"><TR><TD BGCOLOR="white" BORDER="0">b^3ab</TD></TR>
<TR><TD BGCOLOR="white" BORDER="0">*b^3a^2b</TD></TR>
<TR><TD BGCOLOR="white" BORDER="0">bab^2ab</TD></TR>
<TR><TD BGCOLOR="white" BORDER="0">ba^2b</TD></TR>
<TR><TD BGCOLOR="white" BORDER="0">bab^2a^2b</TD></TR>
<TR><TD BGCOLOR="white" BORDER="0">bab</TD></TR>
</TABLE></TD></TR>
</TABLE>>];
}
gap> Print(DotForDrawingDClassOfElement(S,Transformation([2,3,2,4])),1);
digraph  DClassOfElement {
graph [center=yes,ordering=out];
node [shape=plaintext];
edge [color=cornflowerblue,arrowhead=none];
1 [label=<
<TABLE BORDER="0" CELLBORDER="0" CELLPADDING="0" CELLSPACING="0" PORT="1">
<TR><TD BORDER="0"><TABLE CELLSPACING="0"><TR><TD BGCOLOR="white" BORDER="0">*b^2a^2b^2</TD></TR>
<TR><TD BGCOLOR="white" BORDER="0">a^2b^2</TD></TR>
<TR><TD BGCOLOR="white" BORDER="0">ab^2a^2b^2</TD></TR>
<TR><TD BGCOLOR="white" BORDER="0">ab^2ab^2</TD></TR>
<TR><TD BGCOLOR="white" BORDER="0">ab^2</TD></TR>
<TR><TD BGCOLOR="white" BORDER="0">b^2ab^2</TD></TR>
</TABLE></TD><TD BORDER="0"><TABLE CELLSPACING="0"><TR><TD BGCOLOR="white" BORDER="0">a</TD></TR>
<TR><TD BGCOLOR="white" BORDER="0">b^2a</TD></TR>
<TR><TD BGCOLOR="white" BORDER="0">*a^2</TD></TR>
<TR><TD BGCOLOR="white" BORDER="0">b^2a^2</TD></TR>
<TR><TD BGCOLOR="white" BORDER="0">ab^2a</TD></TR>
<TR><TD BGCOLOR="white" BORDER="0">ab^2a^2</TD></TR>
</TABLE></TD><TD BORDER="0"><TABLE CELLSPACING="0"><TR><TD BGCOLOR="white" BORDER="0">ab^3</TD></TR>
<TR><TD BGCOLOR="white" BORDER="0">b^2ab^3</TD></TR>
<TR><TD BGCOLOR="white" BORDER="0">a^2b^3</TD></TR>
<TR><TD BGCOLOR="white" BORDER="0">b^2a^2b^3</TD></TR>
<TR><TD BGCOLOR="white" BORDER="0">ab^2ab^3</TD></TR>
<TR><TD BGCOLOR="white" BORDER="0">ab^2a^2b^3</TD></TR>
</TABLE></TD><TD BORDER="0"><TABLE CELLSPACING="0"><TR><TD BGCOLOR="white" BORDER="0">ab^2ab</TD></TR>
<TR><TD BGCOLOR="white" BORDER="0">ab^2a^2b</TD></TR>
<TR><TD BGCOLOR="white" BORDER="0">b^2ab</TD></TR>
<TR><TD BGCOLOR="white" BORDER="0">ab</TD></TR>
<TR><TD BGCOLOR="white" BORDER="0">b^2a^2b</TD></TR>
<TR><TD BGCOLOR="white" BORDER="0">a^2b</TD></TR>
</TABLE></TD></TR>
<TR><TD BORDER="0"><TABLE CELLSPACING="0"><TR><TD BGCOLOR="white" BORDER="0">bab^2a^2b^2</TD></TR>
<TR><TD BGCOLOR="white" BORDER="0">bab^2</TD></TR>
<TR><TD BGCOLOR="white" BORDER="0">b^3a^2b^2</TD></TR>
<TR><TD BGCOLOR="white" BORDER="0">b^3ab^2</TD></TR>
<TR><TD BGCOLOR="white" BORDER="0">ba^2b^2</TD></TR>
<TR><TD BGCOLOR="white" BORDER="0">bab^2ab^2</TD></TR>
</TABLE></TD><TD BORDER="0"><TABLE CELLSPACING="0"><TR><TD BGCOLOR="white" BORDER="0">ba^2</TD></TR>
<TR><TD BGCOLOR="white" BORDER="0">bab^2a</TD></TR>
<TR><TD BGCOLOR="white" BORDER="0">ba</TD></TR>
<TR><TD BGCOLOR="white" BORDER="0">bab^2a^2</TD></TR>
<TR><TD BGCOLOR="white" BORDER="0">b^3a</TD></TR>
<TR><TD BGCOLOR="white" BORDER="0">b^3a^2</TD></TR>
</TABLE></TD><TD BORDER="0"><TABLE CELLSPACING="0"><TR><TD BGCOLOR="white" BORDER="0">*ba^2b^3</TD></TR>
<TR><TD BGCOLOR="white" BORDER="0">bab^2ab^3</TD></TR>
<TR><TD BGCOLOR="white" BORDER="0">bab^3</TD></TR>
<TR><TD BGCOLOR="white" BORDER="0">bab^2a^2b^3</TD></TR>
<TR><TD BGCOLOR="white" BORDER="0">b^3ab^3</TD></TR>
<TR><TD BGCOLOR="white" BORDER="0">b^3a^2b^3</TD></TR>
</TABLE></TD><TD BORDER="0"><TABLE CELLSPACING="0"><TR><TD BGCOLOR="white" BORDER="0">b^3ab</TD></TR>
<TR><TD BGCOLOR="white" BORDER="0">*b^3a^2b</TD></TR>
<TR><TD BGCOLOR="white" BORDER="0">bab^2ab</TD></TR>
<TR><TD BGCOLOR="white" BORDER="0">ba^2b</TD></TR>
<TR><TD BGCOLOR="white" BORDER="0">bab^2a^2b</TD></TR>
<TR><TD BGCOLOR="white" BORDER="0">bab</TD></TR>
</TABLE></TD></TR>
</TABLE>>];
}
1
\end{verbatim}
}
This output could be rendered with \textsc{graphiviz}\footnote{\url{https://graphviz.org/}}
by the following commands:
\texttt{dot -Tpng <inputfile> -o <outputfile>}, which will produce png files as output.
Note that, if used on Windows, after installation, the command \texttt{dot -c}
has to be executed to install additional plugins, for example for rendering.
If everything went fine, the following output should be produced.

\noindent\begin{minipage}{0.5\textwidth}
\includegraphics[width=1.03\textwidth]{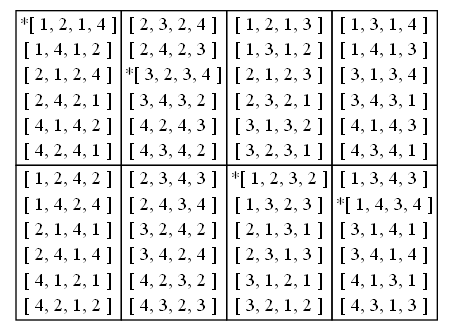}
\end{minipage}%
\begin{minipage}{0.5\textwidth}
\includegraphics[width=1\textwidth]{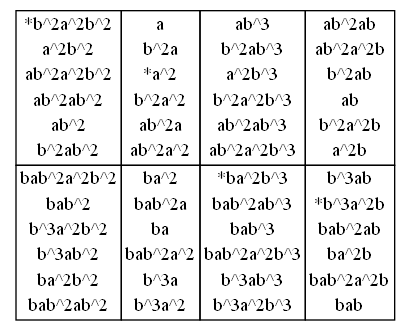}
\end{minipage}
\end{toappendix}

\section{Automata over Binary Alphabets}
\label{sec:binary_case}

Here, we take a closer look at automata over a binary alphabet.
We apply our results and solve an open problem posed in~\cite{DBLP:journals/ijfcs/Maslennikova19}.
In general, if a letter has rank $k$
and some subset is mapped to a subset of size $k$, we must hit the full image of 
this letter. This gives, if we only have two letters but more than two states and no letter has full rank, 
that we can only reach at most two subsets of size $n - 1$. 
So, if more $(n-1)$-sets are reachable, we must have precisely one letter of rank $n - 1$
and Corollary~\ref{cor:n-1_reachable_iff_transitive_perm_grp}
gives the next result.

%
%
%
%


\begin{lemma} 
\label{prop:binary_compl_reach} 
 Let $\Sigma = \{a,b\}$ be a binary alphabet
 and $\mathscr A = (\Sigma, Q, \delta)$
 a finite semi-automaton with $n > 2$ states.
 Then, the following conditions are equivalent:
 \begin{enumerate}
 \item every subset of size $n - 1$ is reachable,
 \item exactly one letter
 acts as a cyclic permutation with a single orbit
 and the other letter has rank $n - 1$. 
 \end{enumerate}
 In particular,  over a binary alphabets, completely reachable automata and those whose set of synchronizing
 words has maximal state complexity are circular.
\end{lemma}
 
 \begin{remark}\footnotesize
  In Lemma~\ref{prop:binary_compl_reach}, we need $n > 2$.
  For let $\mathscr A = (\{a,b\}, \{p,q\}, \delta)$
  with $\delta(p, a) = \delta(q, a) = q$
  and $\delta(p, b) = \delta(q, b) = p$.
  Then $\mathscr A$ is completely reachable, but no letter acts as a non-trivial permutation.
 \end{remark}

With Theorem~\ref{prop:bin_max_sc_1}, we can solve an open problem from~\cite{DBLP:journals/ijfcs/Maslennikova19}.
For $n > 5$, define the automata\footnote{I slightly changed the numbering of
the states with respect to the action of the letter $a$
compared to~\cite{DBLP:journals/ijfcs/Maslennikova19}.}
$\mathscr K_n = (\Sigma, [n], \delta)$, introduced 
in~\cite{DBLP:journals/ijfcs/Maslennikova19},
with
\begin{align*} 
 \delta(i, b) & = i + 1 \mbox{ for } i \in \{0,\ldots, n-2\}, \mbox{ and } \delta(n-1, b) = 0; \\
 \delta(i, a) & = i + 1 \mbox{ for } i \in \{1, \ldots,n-3\}, \delta(n-1,a) = 0,\delta(n-2,a) = 1, 
 \delta(0, a) = 3.
\end{align*}
Please see Example~\ref{ex:aut_families_max_sc} 
for an illustration of this automata family. 
In~\cite{DBLP:journals/ijfcs/Maslennikova19}, it was conjectured that $\stc(\Syn(\mathscr K_n)) = 2^n - n$
for every odd $n > 5$. With Theorem~\ref{prop:bin_max_sc_1},
together with Proposition~\ref{prop:don_compl_reach} 
and Lemma~\ref{lem:compl_reach_implies_max_Syn_2sets}, we can confirm this.

\begin{proposition}
\label{prop:Kn}
  Let $n > 5$ be odd. Then we have $\stc(\Syn(\mathscr K_n)) = 2^n - n$.
\end{proposition}
\begin{proof}
  First, we will show, using Proposition~\ref{prop:don_compl_reach}, that
  the automata $\mathscr K_n$, for odd $n > 5$, are completely reachable.
  Then, we will show, using Proposition~\ref{prop:bin_max_sc_1},
  that all $2$-subsets of states are distinguishable in the power automaton $\mathcal P_{\mathscr K_n}$. 
  With
  Lemma~\ref{lem:compl_reach_implies_max_Syn_2sets}, this would
  then give $\stc(\Syn(\mathscr K_n)) = 2^n - n$.
 \begin{enumerate}
 \item For $n > 5$ odd, the automata $\mathscr K_n$ are completely reachable:  We have two letters, the letter $a$ has rank $n-1$
   and the letter $b$ is a cyclic permutation of all the states.
   Also $\delta(Q, a) = Q \setminus \{ n - 1\}$, $\delta^{-1}(3, a) = \{ 0, 2 \}$
   and $\delta(n - 1, b^4) = 3$.
   If $n$ is odd, then $n$ and $4$ are coprime.
   We have listed the prerequisites of Proposition~\ref{prop:don_compl_reach}, 
   hence applying it gives that $\mathscr K_n$ is completely reachable.

 \item For $n > 5$ odd, in $\mathscr K_n$ all $2$-sets are distinguishable
  in $\mathcal P_{\mathscr K_n}$:   Let $q = 0$. Then
   $$
    \delta(q, ba) = 2 = \delta(q, ab^{n - 1}) = \delta(q, ab^{n + 1 - 2}).
   $$
   For $m \in \{2, \ldots, n - 3\}$, we have $\delta(0, b^m) = m$ and
   $$
    \delta(q, b^ma) = m + 1 = \delta(3, b^{m-2}) = \delta(q, ab^{m-2}) = \delta(q, ab^{n + m-2}).
   $$
   The value $m = n - 2$ does not follow the above pattern, but we have
   $\delta(q, b^{n-2}ab^2) = \delta(1, b^2) = 3 = \delta(q, a)$.
   And lastly, for $m = n - 1$, we have 
   $$
    \delta(q, b^{n-1}a) = q = \delta(q, ab^{n-1-2}).
   $$
   So, with $d = 2$ and $r = 2$, for odd $n$, as then $n - 2$ is not divisible by $d$,
   and with $q = 0$, the prerequisites of Theorem~\ref{prop:bin_max_sc_1}
   are fulfilled and give the claim.
 \end{enumerate}
 So, both statements together with Lemma~\ref{lem:compl_reach_implies_max_Syn_2sets}
 yield $\stc(\Syn(\mathscr K_n)) = 2^n - n$.~$\qed$
\end{proof}

Lastly, let us give some additional examples from the literature~\cite{DBLP:conf/mfcs/AnanichevGV10,AnaVolGus2013,DBLP:journals/corr/Maslennikova14,DBLP:journals/ijfcs/Maslennikova19}
for which our results apply.

\begin{example} 
\label{ex:aut_families_max_sc}
 \footnotesize
 Please see Figure~\ref{fig:aut_ex} for the automata families.
 The automata $\mathscr C_n$ gives the \v{C}ern\'y family,
 the automata $\mathscr L_n$, $\mathscr V_n$, $\mathscr F_n$
 and $\mathscr K_n$ were introduced 
 in~\cite{DBLP:conf/mfcs/AnanichevGV10,DBLP:journals/corr/Maslennikova14,DBLP:journals/ijfcs/Maslennikova19}.
 There, except for $\mathscr K_n$, it was established that in each
 case (for $\mathscr F_n$ only if $n$ is odd and $n > 3$) 
 the set of synchronizing words has maximal state complexity.
 Note that our results, namely Theorem~\ref{prop:bin_max_sc_1},
 together with Proposition~\ref{prop:don_compl_reach}
 and Lemma~\ref{lem:compl_reach_implies_max_Syn_2sets}
 also give these results.
 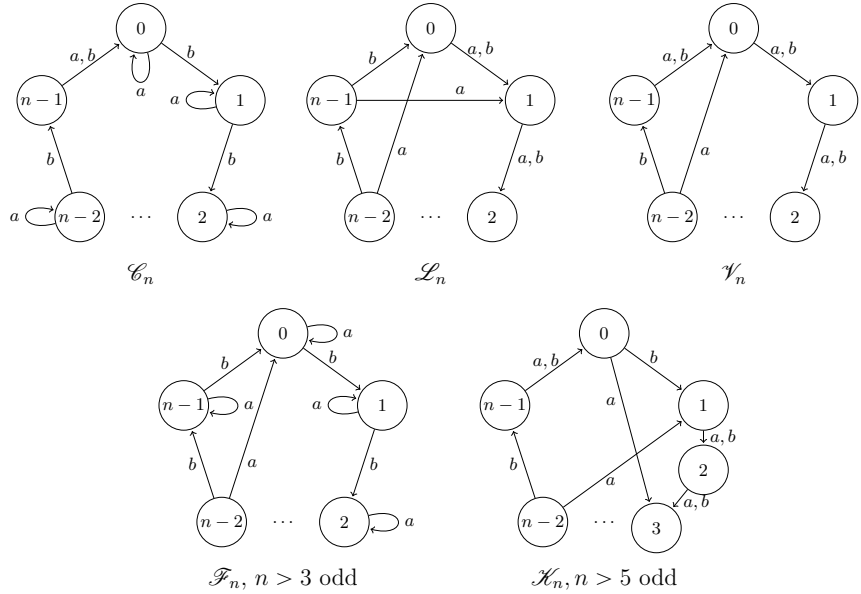
\begin{figure}[htb]
\begin{minipage}{0.33\textwidth}
 \scalebox{.75}{\begin{tikzpicture}[shorten >=1pt,->,node distance=5cm]
  \tikzstyle{vertex}=[circle,draw,minimum size=25pt,inner sep=1pt]
 
  \foreach \name/\angle/\text in {P-n-2/234/n-2, P-n-1/162/n-1, 
                                  P-0/90/0, P-1/18/1, P-2/-54/2}
    \node[vertex] (\name) at (\angle:1.85cm) {\small $\text$};






   \node at (0, -1.5) {$\ldots$};

   \path[->] (P-n-2) edge [left] node {$b$} (P-n-1);
   \path[->] (P-n-1) edge [above,pos=0.35] node {$a,b$} (P-0);
   \path[->] (P-0) edge [above] node {$b$} (P-1);
   \path[->] (P-1) edge [right] node {$b$} (P-2);
   
   \path[->] (P-n-2) edge [loop left] node {$a$} (P-n-2);
   \path[->] (P-0) edge [loop below] node {$a$} (P-0);
   \path[->] (P-1) edge [loop left] node {$a$} (P-1);
   \path[->] (P-2) edge [loop right] node {$a$} (P-2);
   
   \node at (0, -2.5) {\large $\mathscr C_n$};
\end{tikzpicture}}%
\end{minipage}%
\begin{minipage}{0.33\textwidth}
\scalebox{.75}{\begin{tikzpicture}[shorten >=1pt,->,node distance=5cm]
  \tikzstyle{vertex}=[circle,draw,minimum size=25pt,inner sep=1pt]
 
  \foreach \name/\angle/\text in {P-n-2/234/n-2, P-n-1/162/n-1, 
                                  P-0/90/0, P-1/18/1, P-2/-54/2}
    \node[vertex] (\name) at (\angle:1.85cm) {\small $\text$};
   
   \node at (0, -1.5) {$\ldots$};

   \path[->] (P-n-2) edge [left] node {$b$} (P-n-1);
   \path[->] (P-n-1) edge [above,pos=0.35] node {$b$} (P-0);
   \path[->] (P-0) edge [above] node {$a,b$} (P-1);
   \path[->] (P-1) edge [right] node {$a,b$} (P-2);
   
   \path[->] (P-n-2) edge [right,pos=0.3] node {$a$} (P-0);
   \path[->] (P-n-1) edge [above,pos=0.7] node {$a$} (P-1);
   
   \node at (0, -2.5) {\large $\mathscr L_n$};
\end{tikzpicture}}%
\end{minipage}%
\begin{minipage}{0.33\textwidth}
 \scalebox{.75}{\begin{tikzpicture}[shorten >=1pt,->,node distance=5cm]
  \tikzstyle{vertex}=[circle,draw,minimum size=25pt,inner sep=1pt]
 
  \foreach \name/\angle/\text in {P-n-2/234/n-2, P-n-1/162/n-1, 
                                  P-0/90/0, P-1/18/1, P-2/-54/2}
    \node[vertex] (\name) at (\angle:1.85cm) {\small $\text$};
   
   \node at (0, -1.5) {$\ldots$};

   \path[->] (P-n-2) edge [left] node {$b$} (P-n-1);
   \path[->] (P-n-1) edge [above,pos=0.35] node {$a,b$} (P-0);
   \path[->] (P-0) edge [above] node {$a,b$} (P-1);
   \path[->] (P-1) edge [right] node {$a,b$} (P-2);
   
   \path[->] (P-n-2) edge [right,pos=0.3] node {$a$} (P-0);
   
   \node at (0, -2.5) {\large $\mathscr V_n$};
\end{tikzpicture}}
\end{minipage}

\medskip

\begin{minipage}{0.16\textwidth}
 \textcolor{white}{Test}
\end{minipage}
\begin{minipage}{0.35\textwidth}
 \scalebox{.75}{\begin{tikzpicture}[shorten >=1pt,->,node distance=5cm]
  \tikzstyle{vertex}=[circle,draw,minimum size=25pt,inner sep=1pt]
 
  \foreach \name/\angle/\text in {P-n-2/234/n-2, P-n-1/162/n-1, 
                                  P-0/90/0, P-1/18/1, P-2/-54/2}
    \node[vertex] (\name) at (\angle:1.85cm) {\small $\text$};

   \node at (0, -1.5) {$\ldots$};

   \path[->] (P-n-2) edge [left] node {$b$} (P-n-1);
   \path[->] (P-n-1) edge [above,pos=0.35] node {$b$} (P-0);
   \path[->] (P-0) edge [above] node {$b$} (P-1);
   \path[->] (P-1) edge [right] node {$b$} (P-2);
   
   \path[->] (P-n-2) edge [right,pos=0.25] node {$a$} (P-0);
   \path[->] (P-n-1) edge [loop right] node {$a$} (P-n-1);
   \path[->] (P-0) edge [loop right] node {$a$} (P-0);
   \path[->] (P-1) edge [loop left] node {$a$} (P-1);
   \path[->] (P-2) edge [loop right] node {$a$} (P-2);
   
   \node at (0, -2.5) {\large $\mathscr F_n$, $n > 3$ odd};
\end{tikzpicture}}%
\end{minipage}%
\begin{minipage}{0.4\textwidth}
\scalebox{.75}{\begin{tikzpicture}[shorten >=1pt,->,node distance=5cm]
  \tikzstyle{vertex}=[circle,draw,minimum size=25pt,inner sep=1pt]
 
  \foreach \name/\angle/\text in {P-n-2/234/n-2, P-n-1/162/n-1, 
                                  P-0/90/0, P-1/18/1, P-2/-18/2, P-3/-60/3}
    \node[vertex] (\name) at (\angle:1.85cm) {\small $\text$};
   
   \node at (0, -1.5) {$\ldots$};

   \path[->] (P-n-2) edge [left] node {$b$} (P-n-1);
   \path[->] (P-n-1) edge [above,pos=0.35] node {$a,b$} (P-0);
   \path[->] (P-0) edge [above] node {$b$} (P-1);
   \path[->] (P-1) edge [right] node {$a,b$} (P-2);
   \path[->] (P-2) edge [right,pos=0.7] node {$a,b$} (P-3);
   
   \path[->] (P-n-2) edge [right,pos=0.3] node {$a$} (P-1);
   \path[->] (P-0) edge [left,pos=0.3] node {$a$} (P-3);
   
   \node at (0, -2.5) {\large $\mathscr K_n , n > 5 $ odd};
\end{tikzpicture}}%
\end{minipage}%
 
   

   
  
   
 \caption{\footnotesize Families of automata whose sets of synchronizing words have maximal state complexity.
  Please see Example~\ref{ex:aut_families_max_sc} for explanation.}
 \label{fig:aut_ex}
\end{figure}
 


\end{example}

\section{Conclusion}

 
 We have stated sufficient criteria for completely reachable generalized circular automata
 with a letter of rank $n - 1$ to deduce that their set of synchronizing
 words has maximal state complexity. 
 Note that by our results, every completely
 reachable automaton over a binary alphabet must have this form.
 It is natural to ask if we can generalize this to obtain a sufficient
 and necessary criterion.
 As a step in this direction, another family for which this might be tackled first
 is the family of circular automata $\mathscr A = (\Sigma, Q, \delta)$ over a binary alphabet
 with a rank $n - 1$ letter $a$ such that $\delta(Q, aa) = \delta(Q, a)$, i.e.,
 the set $\delta(Q, a)$ is permuted by $a$.
 For these automata, we can find a power of $a$ such that $a$
 acts as the identity on $\delta(Q, a)$ and the single state in $Q \setminus \delta(Q, a)$
 is mapped into $\delta(Q, a)$. Note that these automata closely resemble\todo{class oder family, einheitlich sagen!}
 those of the \Cerny~family.
 Hence, we can find easy sufficient criteria for these automata
 by applying our obtained results to the resulting automaton, where the power
 of $a$ is considered as the new rank $n - 1$ letter.
 However, as Example~\ref{ex:thm_not_usable} shows, such a criterion  
 is also not necessary. But a more finer analysis might give
 sufficient and necessary criteria.



{\footnotesize 
\bibliographystyle{splncs04}
\bibliography{ms} 
}


\end{document}